\documentclass[12pt,a4paper]{article}
\usepackage[T1]{fontenc}
\usepackage[utf8]{inputenc}
\usepackage{a4wide}
\usepackage{amsmath}
\usepackage{amsfonts}
\usepackage{amsthm}
\usepackage{hyperref}

\newtheorem{theorem}{Theorem}
\newcommand{\LL}{{\mathrm{L}}}
\newcommand{\dom}{{\mathcal{D}}}

\begin{document}

\title{A new version of the Aharonov-Bohm effect}
\author{C\'esar R. de Oliveira$^1$\;\; {\small and} \;\; Renan G. Romano$^2$\\ 
\small
\em $^1$Departamento de Matem\'{a}tica -- UFSCar, \small \it S\~{a}o Carlos, SP, 13560-970 Brazil\\
\small\it $^2$Departamento de Matem\'{a}tica -- UFSC, Blumenau, SC, 89036-004 Brazil\\}
\date{January 13, 2020}

\maketitle

\begin{abstract} 
We propose a simple situation in which the magnetic Aharonov-Bohm potential influences the values of the deficiency indices of the initial Schr\"{o}dinger operator, so determining whether the particle interacts with the solenoid or not. Even with the particle excluded from the magnetic field, the number of self-adjoint extensions of the initial Hamiltonian depends on the magnetic flux. This is a new point of view of the Aharonov-Bohm effect.
\end{abstract}

\

\

\section{Introduction}

Yakir Aharonov and David Bohm~\cite{AharonovBohm1959} have  argued that the (magnetic) potential can influence the physics of a quantum particle, even when the particle does not directly interact with the magnetic field  (see~\cite{ES} for an earlier, although more restricted, discussion).  Such kind of influence has been verified in scattering operators and  cross sections (see \cite{AharonovBohm1959,ruijser,dabrowskiStovicek,adamiTeta,rouxYaf,BallesterosWeder2009,deOliveiraPereira2010} to mention just some possible references) and, in some cases, in the Hamiltonian eigenvalues \cite{peshkin1,helffer, dOR}. Of course such influence can be detected in other quantities derived from eigenvalues and/or scattering operators. The ultimate experimental tests are due to Tonomura's group~\cite{Tonomura2,Tonomura1989} and Batelaan and collaborators~\cite{caprezBB,beckerB}.  

The general phenomenon in the magnetic Aharonov-Bohm (AB) effect may be summarized as follows. A magnetic field is confined to an impenetrable region, but it generates a nonzero vector potential outside this region which influences a quantum charged particle and is  usually probed through the  magnetic flux~$\kappa$. The border of this region is the solenoid. The particle motion is restricted to such magnetic  field free region but scattering quantities or eigenvalues may actually depend on~$\kappa$. 

In another direction, there is the problem of  the choice of the boundary condition at the solenoid. Such choice corresponds, mathematically, to  the choice of a self-adjoint extension of the initial magnetic Hamiltonian, since usually the preliminary physical information provides only a symmetric operator;  the self-adjoint  property is required (in fact equivalent)  for a unitary time evolution (conservation of probability) in quantum dynamics.  The topic is related to the existence of anomalies~\cite{BOA},   the presence of anisotropic scale invariances~\cite{brattan}, different surface spectra of Weyl semimetals~\cite{SV},  creation of a pointlike source in quantum field  theory~\cite{ZCK}, studies of topological quantum phases~\cite{AOS} and models in quantum gravity whose time evolution depend on boundary conditions at the origin~\cite{maeda}, to mention only a handful of examples that illustrate the well-known fact that different self-adjoint extensions  correspond to different physics.    An instructive discussion about self-adjoint extensions of the simple case of a particle in a potential well appears in~\cite{BFV}.

By applying a natural shielding method to the initial AB Hamiltonian \cite{kre,magniVG,deOliveiraPereira2008,deOliveiraPereira2011}, the Dirichlet boundary condition (i.e., wavefunctions vanish at the solenoid) is selected. In~\cite{dOR} we have recently proposed a modification of the AB Hamiltonian that is essentially self-adjoint, that is, a model for which there is exactly one self-adjoint extension;  the physical interpretation is the absence of contact of the particle with the solenoid in this case. This was obtained by the introduction of an additional potential that conveniently diverges close to the solenoid.  For circular solenoids, the  presence of the AB effect in this setting was confirmed through the variation of the first eigenvalue with respect to the magnetic flux.

The purpose of this note is to present a situation where, effectively, instead of  scattering quantities or the eigenvalues of the Hamiltonian, the Aharonov-Bohm effect is probed  by the number of self-adjoint extensions of the original energy operator. The basic point is to add a  potential, similar to our approach in~\cite{dOR}, but for solenoids of zero radius (such solenoids have been considered in the original paper by Aharonov and Bohm and in many other works; for instance in~\cite{adamiTeta,dabrowskiStovicek}). Although technically very simple, we have found a situation in which the AB magnetic potential is responsible for the values of the deficiency indices of the original Hamiltonian, a very strong influence that (to the best of our knowledge) has not been reported yet.  This is a new form of the Aharonov-Bohm effect and we think it is worth sharing.

Theorems~\ref{teo:contagem} and~\ref{teo:teo2} present a parameterization of the deficiency indices of the proposed model in terms of the magnetic flux; then the self-adjoint extensions are characterized through the standard von Neumann prescription. Some thought-experiments are proposed in Section~\ref{sectTE}. Concluding discussions appear in Section~\ref{sectConclusions}.

We close this Introduction with some words on philosophical aspects of this work. We have a system whose deficiency indices, and so the parametrization of  the (and how many)   self-adjoint extensions, depend on the magnetic potential; and by varying parameters, we may have any value for such indices (see Theorem~\ref{teo:teo2}). We consider that this adds appealing information to the foundations of Quantum Mechanics, and such results are obtained in a mathematically rigorous approach. 

Another point is that the considered model makes use of the idealized zero radius solenoid, and this was (technically) important for us (compare with~\cite{dOR} that has considered nonzero radius solenoids). Although there are some worries about idealizations, in many instances they make some models mathematically tractable, whereas still keeping important aspects of real systems and help us to understand the foundations of the physical theory; we reinforce that the original scattering calculations by Aharonov and Bohm have also used solenoids of zero radius. A recent interesting discussion about idealizations in the AB effect has been published by Earman~\cite{earman}.

This work is an example supporting the following view of Batelaan and Tonomura~\cite{BT}: ``The AB effect was already implicit in the 1926 Schr\"odinger equation, but it would be another three decades before theorists Yakir Aharonov and David Bohm pointed it out. And to this day, the investigation and exploitation of the AB effect remain far from finished.''

\section{Model and results}\label{sectMR}

Consider the Aharonov-Bohm model of an infinitely long solenoid of zero radius and the dynamics restricted to a plane orthogonal to the solenoid, and with scaled magnetic flux~$\kappa$. The (particle) accessible region is the punctured plane $\Omega:=\mathbb{R}^2\setminus\{0\}$, in the variables $(x,y)$,  and denote $r=\sqrt{x^{2}+y^{2}}$. The standard AB magnetic potential in~$\Omega$ is
\begin{equation}\label{eq:Akappa}
	\mathbf{A}_{\kappa}\left(x,y\right):=\frac{\kappa}{r^{2}}\left(-y,x\right).
\end{equation}

We also consider  an additional scalar (radial) potential~$V(r)$ in~$\Omega$ given by
\begin{equation*}
	V_{\alpha}\left(r\right):=\frac{\alpha}{r^{2}}\,,
\end{equation*}
with fixed $\alpha\in\mathbb{R}$; both $V_\alpha$ and $\mathbf A_\kappa$ are smooth and unbounded in~$\Omega$.  In suitable units, our initial Hamiltonian operator is
\begin{equation}\label{eq:ABdefmeu}
	 H_{\kappa,\alpha}\varphi:=\left(i\nabla+\mathbf{A}_{\kappa}\right)^{2}\varphi+V_{\alpha}\varphi,
\end{equation}
with domain given by the usual smooth functions of compact support $\varphi\in \mathrm{C}_{0}^{\infty}\left(\Omega\right)$. Note that $\alpha=0$ gives the AB initial Hamiltonian as considered in the original Aharonov-Bohm work~\cite{AharonovBohm1959}.

Using the decomposition of the Hilbert space $\LL^{2}(\Omega)=\LL^{2}(\mathbb{R}^{+},r\mathrm{dr})\otimes \LL^{2}(S^{1})$ with respect to angular momentum, the unitary transformation $W:\LL^{2}(\mathbb{R}^{+},r\mathrm{dr})\rightarrow \LL^{2}(\mathbb{R}^{+},\mathrm{dr})$ given by $Wf(r)=r^{\frac{1}{2}}f(r)$, where $\mathbb{R}^{+}=(0,\infty)$, and the completeness of $e^{im\theta}/\sqrt{2\pi}$, $m\in\mathbb{Z}$ in $\LL^{2}(S^{1})$, we can write
\begin{equation*}
	\LL^{2}(\Omega)=\bigoplus_{m=-\infty}^{\infty}(W^{-1}\LL^{2}(\mathbb{R}^{+}))\oplus\left[\frac{e^{im\theta}}{\sqrt{2\pi}}\right],
\end{equation*}
where $[w]$ denotes de linear span of the vector $w$. The above decomposition of the space $\LL^{2}(\Omega)$  induces a decomposition of the operator $H_{\kappa,\alpha}$ in the form
\begin{equation*}
	H_{\kappa,\alpha}=\bigoplus_{m=-\infty}^{+\infty}(W^{-1}H_{\kappa,\alpha}^{(m)}W\otimes\mathbf{1}),
\end{equation*}
where $\mathbf{1}$ is the identity operator and $H_{\kappa,\alpha}^{(m)}$ is defined, for each $m\in\mathbb{Z}$, by
\begin{equation}\label{eq:liouville}
	{H}_{\kappa,\alpha}^{(m)}:=-\dfrac{d^{2}}{dr^{2}}+\dfrac{1}{r^{2}}\left[(m-\kappa)^{2}-\dfrac{1}{4}\right]+\dfrac{\alpha}{r^{2}}\,,
\end{equation}
with domain $\mathrm{C}_{0}^{\infty}(0,\infty)$. Denoting by $T^{*}$ the adjoint of a linear operator $T$, recall that the deficiency subspaces of~$T$ are given by  $K_{\pm}(T):=\mathrm{Ker}\,(T^*\pm i\mathbf 1)$. The above decomposition  implies that the deficiency $K_{\pm}=K_\pm(H_{\kappa,\alpha})$ subspace  equations, 
\[
H_{\kappa,\alpha}^{*}\psi\pm i\psi=0\,,
\] associated with the deficiency indices $n_\pm(H_{\kappa,\alpha})$ of the operator $H_{\kappa,\alpha}$, has a solution $\psi$ if, and only if,
\begin{equation*}
	\psi(r,\theta)=\sum_{m\in\mathbb{Z}}r^{-\frac{1}{2}}\psi_{m}(r)e^{im\theta},
\end{equation*}
where $\psi_{m}$ is a solution to the $K_{\pm}$ equation associated with the operator $H_{\kappa,\alpha}^{(m)}$, for each $m\in\mathbb{Z}$. By the usual limit point and limit circle criteria, it readily follows that the operator $H_{\kappa,\alpha}$ is essentially self-adjoint for all $\alpha\geq 1$.

\begin{theorem}\label{teo:contagem}
	Given $\alpha<1$, the deficiency indices $n_{\pm}$ of the operator $H_{\kappa,\alpha}$ coincides with the number of integers $m\in\mathbb{Z}$ such that
	\begin{equation}\label{eq:contagemZ}
		\left|m-\kappa\right|<\sqrt{1-\alpha}.
	\end{equation}
	In particular, if $\kappa-\tilde{\kappa}\in\mathbb{Z}$, then $n_{\pm}(H_{\kappa,\alpha})=n_{\pm}(H_{\tilde{\kappa},\alpha})$.
\end{theorem}

\begin{proof}
	A direct calculation shows that, by limit point/circle criterion, for each~$m$, the operator $H_{\kappa,\alpha}^{(m)}$ has deficiency indices $n_{\pm}=0$ or $n_{\pm}=1$. Moreover, $n_{\pm}=1$ if, and only if, the parameter~$\alpha$ in the scalar potential in the operator~\eqref{eq:liouville} satisfies
	\begin{equation*}
		\dfrac{(m-\kappa)^{2}+\alpha-\frac{1}{4}}{r^{2}}<\dfrac{3}{4r^{2}}\,.
	\end{equation*}
	This is equivalent to $\left|m-\kappa\right|<\sqrt{1-\alpha}$. Then, for each $m\in\mathbb{Z}$ that satisfies this inequality, we have a unique normalized solution $\psi_{m}^{\pm}$ to the $K_{\pm}$ equation associated with $H_{\kappa,\alpha}^{(m)}$ and a correspondent $r^{\frac{1}{2}}\psi_{m}^{\pm}(r)e^{im\theta}$ solution to the $K_{\pm}(H_{\kappa,\alpha})$ equation. Since these are all the solutions to the $K_{\pm}(H_{\kappa,\alpha})$ equations, we conclude the result.
\end{proof}

 For each  integer $p\geq0$, denote by $a_{p}:=1-(p/2)^{2}$ and $b_{p}:=1-((p+1)/2)^{2}$.  This notation is used in Theorem~\ref{teo:teo2}.
 
\begin{theorem}\label{teo:teo2}
	Let $\alpha<1$ be a fixed real number, $n_{\pm}(H_{\kappa,\alpha})$ the deficiency indices of the operator $H_{\kappa,\alpha}$ and let $p\geq0$ be an integer such that $\alpha\in(b_{p},a_{p}]$. Then the values of the deficiency indices of $H_{\kappa,\alpha}$, in terms of the magnetic flux~$\kappa$, are given by:
	\begin{itemize}
		\item[i)] If $\alpha=a_{p}$ with even $p>0$, then $n_{\pm}(H_{\kappa,\alpha})=p-1$, for all $\kappa\in\mathbb{Z}$ and $n_{\pm}(H_{\kappa,\alpha})=p$ otherwise.

		\item[ii)] If $\alpha=a_{p}$ with odd $p$, then $n_{\pm}(H_{\kappa,\alpha})=p-1$ for all $\kappa\in\frac{1}{2}+\mathbb{Z}$ and $n_{\pm}(H_{\kappa,\alpha})=p$ otherwise.

		\item[iii)] If $\alpha\in(b_{p},a_{p})$ with even $p$, then there exists $\kappa_{0}\in(0,\frac{1}{2})$ such that $n_{\pm}(H_{\kappa,\alpha})=p+1$ for all $\kappa\in(-\kappa_{0},\kappa_{0})+\mathbb{Z}$ and $n_{\pm}(H_{\kappa,\alpha})=p$ otherwise.

		\item[iv)] If $\alpha\in(b_{p},a_{p})$ with odd $p$, then there exists $\kappa_{0}\in(0,\frac{1}{2})$ such that $n_{\pm}(H_{\kappa,\alpha})=p$ for all $\kappa\in[-\kappa_{0},\kappa_{0}]+\mathbb{Z}$ and $n_{\pm}(H_{\kappa,\alpha})=p+1$ otherwise.
	\end{itemize}
\end{theorem}

\begin{proof}
	We will prove item $iii)$; the other proofs follow by similar arguments. Set $\kappa_{0}=\sqrt{1-\alpha}-p/2\in(0,1/2)$ and let $\kappa\in(-\kappa_{0},\kappa_{0})$. By Theorem \ref{teo:contagem}, the deficiency indices $n_{\pm}(H_{\kappa,\alpha})$ equal the number of integers $m$ such that $\left|m-\kappa\right|<\sqrt{1-\alpha}$. Note that this is equivalent to
	\begin{equation}\label{eq:inequality}
		\left(-\frac{p+1}{2}+\kappa,-\frac{p}{2}+\kappa\right)\ni-\sqrt{1-\alpha}+\kappa\;<\;m\;<\;\sqrt{1-\alpha}+\kappa\in\left(\frac{p}{2}+\kappa,\frac{p+1}{2}+\kappa\right).
	\end{equation}
		Since $\left|\kappa\right|<\kappa_{0}=\sqrt{1-\alpha}-\dfrac{p}{2}$ and $p$ is even, the integer numbers~$m$ that satisfy inequality~\eqref{eq:contagemZ} are $m=0,\pm 1,\pm 2,\ldots,\pm p/2$ and, therefore, $n_{\pm}(H_{\kappa,\alpha})=p+1$. To complete this case, note that if $\tilde{\kappa}\in(-\kappa_{0},\kappa_{0})+\mathbb{Z}$, then there exists $\kappa\in(-\kappa_{0},\kappa_{0})$ such that $(\tilde{\kappa}-\kappa)\in\mathbb{Z}$ and, by Theorem~\ref{teo:contagem}, we have $n_{\pm}(H_{\tilde{\kappa},\alpha})=n_{\pm}(H_{\kappa,\alpha})=p+1$.

		Now, let $\kappa_{0}\leq\left|\kappa\right|\leq \frac{1}{2}$. If $\kappa>0$, then
		\begin{equation*}
			-\frac{p}{2}=-\sqrt{1-\alpha}+\kappa_{0}\;\leq\;-\sqrt{1-\alpha}+\kappa\;<\;-\frac{p}{2}+\frac{1}{2},
		\end{equation*}
		and the integers that satisfy~\eqref{eq:contagemZ} are $m=0,\pm 1,\pm 2,\ldots,\pm ({p}/{2}-1),{p}/{2}$. If $\kappa<0$, then
		\begin{equation*}
			\frac{p}{2}-\frac{1}{2}\;\leq\;\sqrt{1-\alpha}-\kappa\;\leq\;\sqrt{1-\alpha}-\kappa_{0}=\frac{p}{2},
		\end{equation*}
		and the allowed integer numbers are $m=\pm 0,\pm1, \pm2,\ldots,\pm({p}/{2}-1),-{p}/{2}$. In any case, $n_{\pm}(H_{\kappa,\alpha})=p$. As before, if $\tilde{\kappa}\notin(-\kappa_{0},\kappa_{0})+\mathbb{Z}$, then there exists $\kappa\in(-1/2,-\kappa_{0}]\cup[\kappa_{0},1/2)$ such that $(\tilde{\kappa}-\kappa)\in\mathbb{Z}$ and again $n_{\pm}(H_{\tilde{\kappa},\alpha})=n_{\pm}(H_{\kappa,\alpha})=p$.
\end{proof}

Now that we have the deficiency indices of the operator $H_{\kappa,\alpha}$, we can parametrize the self-adjoint extensions by unitary matrices~\cite{AkhiezerGlazman1993}. To do so, consider the $K_{\pm}(H_{\kappa,\alpha}^{(m)})$ equation, which reads
\begin{equation*}
	\dfrac{d^{2}\psi}{dr^{2}}+\left[-(\pm i)-\dfrac{(m-\kappa)^{2}+\alpha-\frac{1}{4}}{r^{2}}\right]\psi=0.
\end{equation*}

Using modified Bessel functions, we arrived at the solution

\begin{equation*}
	\psi_{\pm}(r):=r^{\frac{1}{2}}e^{\nu\pi i}K_{\nu}(e^{\mp\frac{\pi}{4}i}r),
\end{equation*}
where $\nu=(m-\kappa)^{2}+\alpha$ and $K_{\nu}$ is the McDonald function of order $\nu$. Note that the analysis of the asymptotic of $K_{\nu}$ (see \cite{abramowitz1964})
\begin{equation*}
	K_{\nu}(z)\sim\frac{1}{2}\Gamma(\nu)\,\big(\frac{z}{2}\big)^{-\nu}
\end{equation*}
implies that the solution $\psi_{\pm}(r)$ is in $\LL^{2}(0,\infty)$ if, and only if, $\nu=(m-\kappa)^{2}+\alpha<1$, which is equivalent to the result in Theorem \ref{teo:contagem}.

Using theses facts and the considerations at the beginning of this section, we see that the deficiency subspaces $K_{\pm}$ of $H_{\kappa,\alpha}$ have an analytic (with respect to $\nu$) basis given by
\begin{equation*}
	\psi^{(m)}_{\pm}(z)=e^{\nu\pi i}K_{\nu}(e^{\mp\frac{\pi}{4}i}r)e^{im\theta},
\end{equation*}
for $m\in\mathbb{Z}$ such that $-\sqrt{1-\alpha}+\kappa<m<\sqrt{1-\alpha}+\kappa$. Denoting by $U$ a unitary map between the finite dimension vector spaces $K_{+}(H_{\kappa,\alpha})$ and $K_{-}(H_{\kappa,\alpha})$ and applying the von Neumann approach to self-adjoint extensions~\cite{AkhiezerGlazman1993}, one finds that such extensions $H_{\kappa,\alpha}^{U}$ are explicitly given by
\begin{equation}\label{eq:domain}
	\begin{split}
		\dom(H_{\kappa,\alpha}^{U})&=\left\{u=v+\psi_{+}-U\psi_{+}\mid u\in \dom (\overline{H}_{\kappa,\alpha}), \psi_+\in K_{+}(H_{\kappa,\alpha})\right\},\\
		H_{\kappa,\alpha}^{U}u&=H_{\kappa,\alpha}^{*}u=\overline{H}_{\kappa,\alpha}v-i\psi_{+}-iU\psi_{+},
	\end{split}
\end{equation}
where, for $\Gamma=\{m\in\mathbb{Z}:-\sqrt{1-\alpha}+\kappa<m<\sqrt{1-\alpha}+\kappa\}$,
\begin{equation}\label{eq:general}
	\begin{split}
		v\in \dom(\overline{H_{\kappa,\alpha}}),\quad \psi_{+}=\sum_{m\in\Gamma}c_{m}\psi_{+}^{(m)},\quad c_{j}\in\mathbb{C},\\
		U\psi_{+}=\sum_{m\in\Gamma}\tilde{c}_{m}\psi_{-}^{(m)},\quad \tilde{c}_{j}=\sum_{l\in\Gamma}\tilde{U}_{jl}c_{l},\quad j\in\Gamma,
	\end{split}
\end{equation}
and $\tilde{U}$ is an element of the unitary group $U(d_{\kappa})$ of $d_{\kappa}\times d_{\kappa}$ unitary complex matrices for $d_{\kappa}=n_{\pm}(H_{\kappa,\alpha})$ and $\overline{T}$ denotes the operator closure of~$T$.

\section{Some thought-experiments}\label{sectTE}

In this section we try to propose some  experiments for which our analysis above could be relevant. The main difficulty in proposing an experiment where such different self-adjoint extensions could play an explicit role is the presence of the additional potential $V_{\alpha}\left(r\right)=\frac{\alpha}{r^{2}}$. We first note that this is  a standard potential to build models with different self-adjoint extensions, including the theoretical speculations about the role of different extensions, and also a source of controversy in the physics literature; see, for instance,  \cite{MRMYSB, CEFGC,BMG,GR,AACMRNYSB} and Chapter~7 of the book~\cite{GTV} (and references therein). 

In the following we describe some possible situations modelled by operators directly related to the proposed operator action~\eqref{eq:ABdefmeu}; in prospective real experiments, the following proposals will certainly be enhanced by the backgrounds and skills of experimentalists. 

\begin{enumerate}

\item The AB effect has an application in atomic interferometry with neutral atoms, and the authors of~\cite{WHW}  have proposed an interesting experiment (see also~\cite{wilkens,ASV}). The radial electric field~$E$ of a homogeneous straight charged wire, with charge density~$\varrho$, polarizes a scattered neutral atom with velocity~$v$; then a uniform magnetic field~$B$ is applied parallel to the wire. The atom with mass~$M'$ moving in these two fields will obtain an AB phase. If terms of order $v^2/c^2$ are neglected ($c$~is the speed of light), the Lagrangian for the atom is given by 
\[
L=\frac{1}{2} M v^{2}+\gamma (B \times E) \, \cdot \frac{v}{c}+\frac{1}{2} \gamma E^{2},
\]where $M=M'+\gamma B^2/c^2$ and~$\gamma$ is the electric polarizability. The corresponding Schr\"odinger operator has the form~\eqref{eq:ABdefmeu} with 
\[
\kappa=\frac{\gamma\varrho}{2\pi c}
\] (see equation~(2) in~\cite{ASV}, which appears in polar coordinates and without the action of the unitary transformation $Wf(r)=r^{\frac{1}{2}}f(r)$; we have also relabeled some parameters to differ from our~$\kappa$ and~$\alpha$). Hence, different self-adjoint extensions may intervene and, according to the results in Section~\ref{sectMR}, the  deficiency numbers may be controlled by selecting the electric polarizability and/or the wire charge density.

\item  In~\cite{ASV}, another set-up was proposed that is also modelled by the operator~\eqref{eq:ABdefmeu}.  It is the motion of a particle in the usual AB potential which is superimposed by a  scalar potential
\[
U(r) = \frac{\alpha}{ r^2}
\] (i.e., our $V_\alpha$) of an electrically charged string in the direction of the solenoid. See equations~(3) and~(5) in~\cite{ASV} (in case their $n\cdot r=0$). 

\item  A rather idealized situation  is the combination of a thin solenoid along the $z$-axis with an electric point charge~$Q$ at the origin of coordinates, and the scattered particle in the $xy$-plane is an electric dipole~$d$, pointed in the $z$-direction and with nonzero net charge. The electric potential due to~$Q$ at~$d$ has approximately two components, the Coulomb $\beta/r$ and the dipole one $\alpha/r^2$ (recall that if $E(r)$ is the electric field on a charged particle, then on a dipole it is~$(d\cdot \nabla)E$, which, in the $xy$-plane in cylindrical coordinates, results proportional to $dE(r)/dr$), so that the initial operator is written, in this approximation,  as
\begin{equation}\label{eqABdipole}
	 H_{\kappa,\alpha,\beta}\varphi:=\left(i\nabla+\mathbf{A}_{\kappa}\right)^{2}\varphi+V_{\alpha}\varphi+\frac{\beta}{r}\varphi.
\end{equation}
We have checked that the deficiency indices of $H_{\kappa,\alpha,\beta}$ also varies with~$\kappa$, so that~\eqref{eqABdipole} presents a corresponding AB-like phenomenon as~\eqref{eq:ABdefmeu} and discussed in Theorems~\ref{teo:contagem} and~\ref{teo:teo2}. The main criticisms of~\eqref{eqABdipole} are that for very small~$r$ the dipole approximation could fail, so it is appropriate to think of an ideal dipole,  and that the dipole direction should be kept approximately constant; but we think that, as an initial proposal of a thought-experiment, it may have some value.
\end{enumerate}

\section{Conclusions}\label{sectConclusions}

For the case $3/4<\alpha<1$, one has $\alpha\in(b_{p},a_{p}]$ with $p=0$, and for $\kappa\in(-\sqrt{1-\alpha},\sqrt{1-\alpha})+\mathbb{Z}$ the operator $H_{\kappa,\alpha}$  has deficiency indices $n_{\pm}=1$,  and $n_{\pm}=0$ otherwise. This means that, if $0\leq\kappa<\sqrt{1-\alpha}$, then the operator $H_{\kappa,\alpha}$ has a one parameter family of self-adjoint extensions and a boundary condition is necessary to specify each of them; these boundary conditions are present in the domain of the operator~\eqref{eq:domain}. For $\sqrt{1-\alpha}\leq\kappa\leq 1/2$, the operator has only one self-adjoint extension and no boundary condition is needed; in this case, the self-adjoint extension is the operator closure of the initial Hamiltonian. 

Another interesting case occurs when one considers $\alpha=0$ in~\eqref{eq:ABdefmeu}, so  $\alpha\in(b_{p},a_{p}]$ with $p=2$, and we have the initial Hamiltonian $H_{\kappa}:=H_{\kappa,0}$ considered by Aharonov and Bohm \cite{AharonovBohm1959} (for zero radius) and our results establish that $n_{\pm}(H_{\kappa})=1$, if $\kappa\in\mathbb{Z}$, and $n_{\pm}(H_{\kappa,\alpha})=2$ for all $\kappa\notin\mathbb{Z}$. So, the initial Hamiltonian has a one parameter family of self-adjoint extensions when the magnetic field is turned off and a four parameter family when the magnetic field grows up from zero with $\kappa\notin\mathbb{Z}$. This can be visualized in the  formula for the general self-adjoint extension~\eqref{eq:general} since the unitary group $U(n_{\pm})$ has dimension $n_{\pm}^{2}$ which explicitly depends on~$\kappa$.

Summing up, in the proposed (simple) model the magnetic potential, in a free magnetic field region, not only may  influence  the scattering parameters and/or eigenvalues of  the energy operator, but also dictates  the number of self-adjoint extensions of the initial Hamiltonian operator, in particular determining whether the particle interacts with the solenoid or not and giving rise to different physics. Experimentally, this would (in principle) be observed by how the apparatus is prepared; that is, in case of just one self-adjoint extension all experimental results would be independent of the experimental preparation, whereas in the case with more than one extension these results would depend on the details of the experimental setup. Potential experimental situations are proposed in Section~\ref{sectTE}.

\subsubsection*{Acknowledgments} 
CRdO thanks the partial support by CNPq (a Brazilian government agency, under contract 303503/2018-1).

\end{document}